\pdfoutput=1
\documentclass{scrartcl}

\usepackage[utf8]{inputenc}
\usepackage[T1]{fontenc}
\usepackage{afterpage}
\usepackage{amsmath,amssymb,amsthm}
\usepackage{chemformula}
\usepackage{enumerate}
\usepackage{listings}
\usepackage{lmodern}
\usepackage{url}

\lstdefinelanguage{Maple}%
{morekeywords={and,assuming,break,by,catch,description,do,done,%
    elif,else,end,error,export,fi,finally,for,from,global,if,%
    implies,in,intersect,local,minus,mod,module,next,not,od,%
    option,options,or,proc,quit,read,return,save,stop,subset,then,%
    to,try,union,use,uses,while,xor},%
  sensitive=true,%
  morecomment=[l]\#,%
  morestring=[b]",%
  morestring=[d]"%
}[keywords,comments,strings]

\newcommand{\ACF}{\mathfrak{A}}
\newcommand{\K}{\mathfrak{K}}
\newcommand{\CC}{\mathbb{C}}
\newcommand{\dotequal}{\mathrel{\dot{=}}}
\newcommand{\LR}{\mathcal{L}}
\newcommand{\FF}{\mathbb{F}}
\newcommand{\KK}{\mathbb{K}}
\newcommand{\NN}{\mathbb{N}}
\newcommand{\QQ}{\mathbb{Q}}
\newcommand{\RR}{\mathbb{R}}
\newcommand{\ZZ}{\mathbb{Z}}
\newcommand{\hilb}[1]{\widehat{#1}}

\theoremstyle{plain}
\newtheorem{conj}{Conjecture}
\newtheorem{lemma}[conj]{Lemma}
\newtheorem{proposition}[conj]{Proposition}
\theoremstyle{definition}
\newtheorem{example}[conj]{Example}

\allowdisplaybreaks[3]

\KOMAoptions{
  abstract=true,
  fontsize=10pt,
  DIV=10,
  headinclude=false,
  footinclude=false}

\addtokomafont{author}{\large}

\recalctypearea

\begin{document}
\title{First-Order Tests for Toricity}

\author{Hamid Rahkooy\\
  CNRS, Inria, and the University of Lorraine, Nancy, France\\
  \texttt{hamid.rahkooy@inria.fr}
  \and
  Thomas Sturm\\
  CNRS, Inria, and the University of Lorraine, Nancy, France\\
  MPI Informatics and Saarland University, Saarbrücken, Germany\\
  \texttt{thomas.sturm@loria.fr}}

\maketitle

\begin{abstract}
  Motivated by problems arising with the symbolic analysis of steady state
  ideals in Chemical Reaction Network Theory, we consider the problem of testing
  whether the points in a complex or real variety with non-zero coordinates form
  a coset of a multiplicative group. That property corresponds to Shifted
  Toricity, a recent generalization of toricity of the corresponding polynomial
  ideal. The key idea is to take a geometric view on varieties rather than an
  algebraic view on ideals. Recently, corresponding coset tests have been
  proposed for complex and for real varieties. The former combine numerous
  techniques from commutative algorithmic algebra with Gröbner bases as the
  central algorithmic tool. The latter are based on interpreted first-order
  logic in real closed fields with real quantifier elimination techniques on the
  algorithmic side. Here we take a new logic approach to both theories, complex
  and real, and beyond. Besides alternative algorithms, our approach provides a
  unified view on theories of fields and helps to understand the relevance and
  interconnection of the rich existing literature in the area, which has been
  focusing on complex numbers, while from a scientific point of view the
  (positive) real numbers are clearly the relevant domain in chemical reaction
  network theory. We apply prototypical implementations of our new approach to a
  set of 129 models from the BioModels repository.
\end{abstract}

\section{Introduction}

We are interested in situations where the points with non-zero coordinates in a
given complex or real variety form a multiplicative group or, more generally, a
coset of such a group. For irreducible varieties this corresponds to toricity
\cite{fulton_introduction_2016,eisenbud-sturmfels-binomials} and shifted
toricity \cite{grigoriev_milman2012, GrigorievIosif:19a}, respectively, of both
the varieties and the corresponding ideals.

While toric varieties are well established and have an important role in
algebraic geometry \cite{fulton_introduction_2016,eisenbud-sturmfels-binomials},
our principal motivation here to study generalizations of toricity comes from
the sciences, specifically \textit{chemical reaction networks} such as the
following model of the kinetics of intra- and intermolecular zymogen activation
with formation of an enzyme-zymogen complex
\cite{doi:10.1111/j.1432-1033.2004.04400.x}, which can also be found as model
no.~92\footnote{\url{https://www.ebi.ac.uk/compneur-srv/biomodels-main/publ-model.do?mid=BIOMD0000000092}}
in the BioModels database \cite{BioModels2015a}:
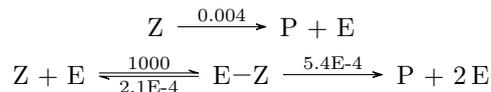
\begin{gather*}
  \ch[label-style=\scriptsize]{Z ->[0.004] P + E}\\
  \ch[label-style=\scriptsize]{Z + E <=>[1000][2.1E-4] E-Z ->[5.4E-4] P + 2 E}
\end{gather*}
Here Z stands for zymogen, P is a peptide, E is an enzyme, E---Z is the enzyme
substrate complex formed from that enzyme and zymogen. The reactions are
labelled with reaction \emph{rate constants}.

Let $x_1$, \dots, $x_4: \RR \to \RR$ denote the \emph{concentrations} over time of the
\emph{species} Z, P, E, E--Z, respectively. Assuming mass action kinetics one can
derive \emph{reaction rates} and furthermore a system of autonomous ordinary
differential equations describing the development of concentrations in the
overall network \cite[Section 2.1.2]{feinberg-book}:
\begin{align*}
  \dot{x_1} &= p_1/100000, & p_1 &= -100000000 x_{1} x_{2} - 400 x_{1} + 21 x_{4},\\
  \dot{x_3} &= p_3/50000, & p_2 &= -100000000 x_{1} x_{2} + 400 x_{1} + 129 x_{4},\\
  \dot{x_2} &= p_2/100000, & p_3 &= 200 x_{1} + 27 x_{4},\\
  \dot{x_4} &= p_4/4000, & p_4 &= 4000000 x_{1} x_{2} - 3 x_{4}.
\end{align*}
The chemical reaction is in \emph{equilibrium} for positive concentrations of
species lying in the real variety of the steady state ideal
\begin{displaymath}
  \langle p_1, \dots, p_4 \rangle \subseteq \ZZ[x_1, \dots, x_4],
\end{displaymath}
intersected with the first orthant of $\RR^4$.

Historically, the principle of \emph{detailed balancing} has attracted
considerable attention in the sciences. It states that at equilibrium every
single reaction must be in equilibrium with its reverse reaction. Detailed
balancing was used by Boltzmann in 1872 in order to prove his H-theorem
\cite{boltzmann1964lectures}, by Einstein in 1916 for his quantum theory of
emission and absorption of radiation \cite{einstein1916strahlungs}, and by
Wegscheider \cite{Wegscheider1901} and Onsager \cite{onsager1931reciprocal} in
the context of \textit{chemical kinetics}, which lead to Onsager's Nobel prize
in Chemistry in 1968. In the field of symbolic computation, Grigoriev and Weber
\cite{GrigorievWeber2012a} applied results on binomial varieties to study
reversible chemical reactions in the case of detailed balancing.

In particular with the assumption of irreversible reactions, like in our
example, detailed balancing has been generalized to \emph{complex balancing}
\cite{feinberg_stability_1987,feinberg-book,Horn1972}, which has widely been
used in the context of chemical reaction networks. Here one considers
\emph{complexes}, like Z, P + E, Z + E, etc.~in our example, and requires for
every such complex that the sum of the reaction rates of its inbound reactions
equals the sum of the reaction rates of its outbound reactions.

Craciun et al.~\cite{craciun_toric_2009} showed that \emph{toric dynamical
  systems} \cite{Feinberg1972,Horn1972}, in turn, generalize complex balancing.
The generalization of the principle of complex balancing to toric dynamical
systems has obtained considerable attention in the last years
\cite{perez_millan_chemical_2012,gatermann_bernsteins_2005,craciun_toric_2009,muller_sign_2016}.
Millan, Dickenstein and Shiu in \cite{perez_millan_chemical_2012} considered
steady state ideals with binomial generators. They presented a sufficient linear
algebra condition on the stoichiometry matrix of a chemical reaction network in
order to test whether the steady state ideal has binomial generators. Conradi
and Kahle showed that the sufficient condition is even equivalent when the ideal
is homogenous \cite{conradi2015detecting,kahle-binomial-package-2010,Kahle2010}.
That condition also led to the introduction of MESSI systems
\cite{millan_structure_2018}. Recently, binomiality of steady states ideals was
used to infer network structure of chemical reaction networks out of measurement
data \cite{Wang_Lin_Sontag_Sorger_2019}.

Besides its scientific adequacy as a generalization of complex balancing there
are practical motivations for studying toricity. Relevant models are typically
quite large. For instance, with our comprehensive computations in this article
we will encounter one system with 90 polynomials in dimension 71. This brings
symbolic computation to its limits. Our best hope is to discover systematic
occurrences of specific structural properties in the models coming from a
specific context, e.g.~the life sciences, and to exploit those structural
properties towards more efficient algorithms. In that course, toricity could
admit tools from toric geometry, e.g., for dimension reduction.

Detecting toricity of varieties in general, and of steady state varieties of
chemical reaction networks in particular, is a difficult problem. The first
issue in this regard is finding suitable notions to describe the structure of
the steady states. Existing work, such as the publications mentioned above,
typically focuses on the complex numbers and addresses algebraic properties of
the steady state ideal, e.g., the existence of binomial Gröbner bases. Only
recently, a group of researchers including the authors of this article have
taken a geometric approach, focusing on varieties rather than ideals
\cite{GrigorievIosif:19a,grigoriev_milman2012}. Besides irreducibility, the
characteristic property for varieties $V$ to be toric over a field $\KK$ is that
$V \cap (\KK^*)^n$ forms a multiplicative group. More generally, one considers
\emph{shifted toricity}, where $V \cap (\KK^*)^n$ forms a coset of a multiplicative
group.

It is important to understand that chemical reaction network theory generally
takes place in the interior of the first orthant of $\RR^n$, i.e., all species
concentrations and reaction rates are assumed to be strictly positive
\cite{feinberg-book}. Considering $(\CC^*)^n$ in contrast to $\CC^n$ resembles the
strictness condition, and considering also $(\RR^*)^n$ in
\cite{GrigorievIosif:19a} was another step in the right direction.

The plan of the article is as follows. In Section~\ref{se:firstorder} we
motivate and formally introduce first-order characterizations for shifted
toricity, which have been used already in \cite{GrigorievIosif:19a}, but
exclusively with real quantifier elimination methods. In
Section~\ref{se:generalization} we put a model theoretic basis and prove
transfer principles for our characterizations throughout various classes of
fields, with zero as well as with positive characteristics. In
Section~\ref{se:hilbert} we employ Hilbert's Nullstellensatz as a decision
procedure for uniform word problems and use logic tests also over algebraically
closed fields. This makes the link between the successful logic approach from
\cite{GrigorievIosif:19a} and the comprehensive existing literature cited above.
Section~\ref{se:complexity} clarifies some asymptotic worst-case complexities
for the sake of scientific rigor. In Section~\ref{se:computations} it turns out
that for a comprehensive benchmark set of 129 models from the BioModels database
\cite{BioModels2015a} quite simple and maintainable code, requiring only
functionality available in most decent computer algebra systems and libraries,
can essentially compete with highly specialized and more complicated purely
algebraic methods. This motivates in Section~\ref{se:conclusions} a perspective
that our symbolic computation approach has a potential to be interesting for
researchers in the life sciences, with communities much larger than our own,
with challenging applications, not least in the health sector.

\section{Syntax: First-Order Formulations of Group and Coset Properties}\label{se:firstorder}
In this section we set up our first-order logic framework. We are going to use
interpreted first-order logic with equality over the signature
$\LR = (0, 1, +, -, \cdot)$ of rings.

For any field $\KK$ we denote its multiplicative group $\KK \setminus \{0\}$ by
$\KK^*$. For a coefficient ring $Z \subseteq \KK$ and
$F \subseteq Z[x_1, \dots, x_n]$ we denote by $V_\KK(F)$, or shortly $V(F)$, the variety of
$F$ over $\KK$. Our signature $\LR$ naturally induces coefficient rings rings
$Z = \ZZ/p$ for finite characteristic $p$, and $Z = \ZZ$ for characteristic $0$. We
define $V(F)^* = V(F) \cap (\KK^*)^n \subseteq (\KK^*)^n$. Note that the direct product
$(\KK^*)^n$ establishes again a multiplicative group.

Let $F = \{f_1, \dots, f_m\} \subseteq Z[x_1, \dots, x_n]$. The following semi-formal conditions
state that $V(F)^*$ establishes a coset of a multiplicative subgroup of
$(\KK^*)^n$:
\begin{gather}
  \forall g, x \in (\KK^*)^n\mathord{:}
    \quad g \in V(F) \land gx \in V(F) \Rightarrow gx^{-1} \in V(F)
    \label{eq:AAAinv}\\[1ex]
  \forall g, x, y \in (\KK^*)^n\mathord{:}
    \quad g \in V(F) \land gx \in V(F) \land gy \in V(F) \Rightarrow gxy \in V(F)
    \label{eq:AAAmult}\\[1ex]
  V(F) \cap (\KK^*)^n \neq \emptyset.
    \label{eq:AAAne}
\end{gather}

If we replace (\ref{eq:AAAne}) with the stronger condition
\begin{equation}
  \label{eq:group}
  1 \in V(F),
\end{equation}
then $V(F)^*$ establishes even a multiplicative subgroup of $(\KK^*)^n$. We allow
ourselves to less formally say that $V(F)^*$ is a coset or group over $\KK$,
respectively.

Denote $M = \{1, \dots, m\}$, $N = \{1, \dots, n\}$, and for
$(i, j) \in M \times N$ let $d_{ij} = \deg_{x_j}(f_i)$. We shortly write
$x = (x_1, \dots, x_n)$, $y=(y_1, \dots, y_n)$, $g=(g_1, \dots, g_n)$. Multiplication
between $x$, $y$, $g$ is coordinate-wise, and
$x^{d_i} = x_1^{d_{i1}} \cdots x_n^{d_{in}}$. As a first-order $\LR$-sentence,
condition (\ref{eq:AAAinv}) yields
\begin{multline*}
  \iota ~\dotequal~ \textstyle
  \forall g_1 \dots \forall g_n \forall x_1 \dots \forall x_n \biggl(
  \bigwedge\limits_{j=1}^n g_j \neq 0 \land
  \bigwedge\limits_{j=1}^n x_j \neq 0 \land{}\\
  \textstyle
  \bigwedge\limits_{i=1}^m f_i(g_1,\dots,g_n) = 0 \land
  \bigwedge\limits_{i=1}^m f_i(g_1x_1,\dots,g_nx_n) = 0
  \longrightarrow \bigwedge\limits_{i=1}^m x^{d_i} f_i(g_1x_1^{-1},\dots,g_nx_n^{-1}) = 0\biggr).
\end{multline*}
Here the multiplications with $x^{d_i}$ drop the principal denominators from
$f_i(g_1x_1^{-1},\dots,g_nx_n^{-1})$. This is an equivalence transformation, because
the left hand side of the implication constrains $x_1$, \dots,~$x_n$ to be different
from zero.

Similarly, condition (\ref{eq:AAAmult}) yields a first-order $\LR$-sentence
\begin{multline*}
  \mu ~\dotequal~ \textstyle
  \forall g_1 \dots \forall g_n
  \forall x_1 \dots \forall x_n
  \forall y_1 \dots \forall y_n
  \biggl(
  \bigwedge\limits_{j=1}^n g_j \neq 0 \land
  \textstyle
  \bigwedge\limits_{j=1}^n x_j \neq 0 \land
  \bigwedge\limits_{j=1}^n y_j \neq 0 \land{}\\
  \textstyle
  \bigwedge\limits_{i=1}^m f_i(g_1,\dots,g_n) = 0 \land
  \bigwedge\limits_{i=1}^m f_i(g_1x_1,\dots,g_nx_n) = 0 \land
  \bigwedge\limits_{i=1}^m f_i(g_1y_1,\dots,g_ny_n) = 0\\
  \textstyle
  \longrightarrow \bigwedge\limits_{i=1}^m f_i(g_1x_1y_1, \dots, g_nx_ny_n) = 0\biggr).
\end{multline*}

For condition (\ref{eq:AAAne}) we consider its logical negation
$V(F) \cap (\KK^*)^n = \emptyset$, which gives us an $\LR$-sentence
\begin{displaymath}
  \eta ~\dotequal~ \textstyle
  \forall x_1 \dots \forall x_n \biggl(
  \bigwedge\limits_{i=1}^m f_i = 0 \longrightarrow \bigvee\limits_{j=1}^n x_j = 0
  \biggr).
\end{displaymath}
Accordingly, the $\LR$-sentence $\lnot\eta$ formally states (\ref{eq:AAAne}).

Finally, condition (\ref{eq:group}) yields a quantifier-free $\LR$-sentence
\begin{displaymath}
  \textstyle
  \gamma ~\dotequal~ \bigwedge\limits_{i=1}^m f_i(1, \dots, 1) = 0.
\end{displaymath}

\section{Semantics: Validity of Our First-Order Formulations over Various Fields}\label{se:generalization}
Let $p \in \NN$ be $0$ or prime. We consider the $\LR$-model classes $\K_p$ of
fields of characteristic $p$ and $\ACF_p \subseteq \K_p$ of algebraically closed fields
of characteristic $p$. Recall that $\ACF_p$ is complete, decidable, and admits
effective quantifier elimination \cite[Note 16]{Tarski:48a}.

We assume without loss of generality that $\LR$-sentences are in prenex normal
form $Q_1 x_1 \dots Q_n x_n \psi$ with $Q_1$, \dots,
$Q_n \in \{\exists, \forall\}$ and $\psi$ quantifier-free. An $\LR$-sentence is called
\emph{universal} if it is of the form $\forall x_1 \dots \forall x_n \psi$ and \emph{existential}
if it is of the form $\exists x_1 \dots \exists x_n \psi$ with $\psi$ quantifier-free. A
quantifier-free $\LR$-sentence is both universal and existential.
\begin{lemma}\label{le:k0iffacf0}
  Let $\varphi$ be a universal $\LR$-sentence. Then
  \begin{displaymath}
    \K_p \models \varphi \quad \text{if and only if} \quad \ACF_p \models \varphi.
  \end{displaymath}
\end{lemma}

\begin{proof}
  The implication from the left to the right immediately follows from
  $\ACF_p \subseteq \K_p$. Assume, vice versa, that $\ACF_p \models \varphi$, and let
  $\KK \in \K_p$. Then $\KK$ has an algebraic closure $\overline \KK \in \ACF_p$, and
  $\overline \KK \models \varphi$ due to the completeness of $\ACF_p$. Since
  $\KK \subseteq \overline \KK$ and $\varphi$ as a universal sentence is persistent under
  substructures, we obtain $\KK \models \varphi$.
\end{proof}

All our first-order conditions $\iota$, $\mu$, $\eta$, and $\gamma$ introduced in the previous
section~\ref{se:firstorder} are universal $\LR$-sentences. Accordingly,
$\lnot \eta$ is equivalent to an existential $\LR$-sentence.

In accordance with the our language $\LR$ we are going to use polynomial
coefficient rings $Z_p = \ZZ/p$ for finite characteristic $p$, and $Z_0 = \ZZ$. Let
$F \subseteq Z_p[x_1,\dots,x_n]$. Then $V(F)^*$ is a coset over $\KK \in \K_p$ if and only if
\begin{equation}
  \KK \models \iota \land \mu \land \lnot \eta.
\end{equation}
Especially, $V(F)^*$ is a group over $\KK$ if even
\begin{equation}
  \KK \models \iota \land \mu \land \gamma,
\end{equation}
where $\gamma$ entails $\lnot \eta$.

\begin{proposition}\label{prop:group}
  Let $F \subseteq Z_p[x_1, \dots, x_n]$, and let $\KK \in \K_p$. Then $V(F)^*$ is a group over
  $\KK$ if and only if at least one of the following conditions holds:
  \begin{enumerate}[(a)]
  \item $\KK' \models \iota \land \mu \land \gamma$ for some $\KK \subseteq \KK' \in \K_p$;
  \item $\KK' \models \iota \land \mu \land \gamma$ for some $\KK' \in \ACF_p$.
  \end{enumerate}
\end{proposition}

\begin{proof}
  Recall that $V(F)^*$ is a group over $\KK$ if and only if
  $\KK \models \iota \land \mu \land \gamma$. If $V(F)^*$ is a group over $\KK$, then (a) holds for
  $\KK' = \KK$. Vice versa, there are two cases. In case (a), we can conclude that
  $\KK \models \iota \land \mu \land \gamma$ because the universal sentence
  $\iota \land \mu \land \gamma$ is persistent under substructures. In case (b), we have
  $\ACF_p \models \iota \land \mu \land \gamma$ by the completeness of that model class. Using
  Lemma~\ref{le:k0iffacf0} we obtain
  $\K_p \models \iota \land \mu \land \gamma$, in particular $\KK \models \iota \land \mu \land \gamma$.
\end{proof}

\begin{example}
  \begin{enumerate}[(i)]
  \item Assume that $V(F)^*$ is a group over $\CC$. Then $V(F)^*$ is a group over
    any field of characteristic $0$. Alternatively, it suffices that $V(F)^*$ is
    a group over the countable algebraic closure $\overline \QQ$ of $\QQ$.
  \item Assume that $V(F)^*$ is a group over the countable field of real
    algebraic numbers, which is not algebraically closed. Then again $V(F)^*$ is
    a group over any field of characteristic $0$.
  \item Let $\varepsilon$ be a positive infinitesimal, and assume that $V(F)^*$ is a group
    over $\RR(\varepsilon)$. Then $V(F)^*$ is group also over $\QQ$ and $\RR$, but not
    necessarily over $\overline{\QQ}$. Notice that $\RR(\varepsilon)$ is not algebraically
    closed.
  \item Assume that $V(F)^*$ is a group over the algebraic closure of $\FF_p$.
    Then $V(F)^*$ is a group over any field of characteristic $p$.
    Alternatively, it suffices that $V(F)^*$ is a group over the algebraic
    closure of the rational function field $\FF_p(t)$, which has been studied
    with respect to effective computations \cite{Kedlaya2006}.
  \end{enumerate}
\end{example}

\begin{proposition}\label{prop:coset}
  Let $F \subseteq Z_p[x_1, \dots, x_n]$ and let $\KK \in \K_p$. Then $V(F)^*$ is a coset over
  $\KK$ if and only if $\KK \models \lnot \eta$ and at least one of the following conditions
  holds:
  \begin{enumerate}[(a)]
  \item $\KK' \models \iota \land \mu$ for some $\KK \subseteq \KK' \in \K_p$;
  \item $\KK' \models \iota \land \mu$ for some $\KK' \in \ACF_p$.
  \end{enumerate}
\end{proposition}

\begin{proof}
  Recall that $V(F)^*$ is a coset over $\KK$ if and only if
  $\KK \models \iota \land \mu \land \lnot \eta$. If $V(F)^*$ is a coset over
  $\KK$, then $\KK \models \lnot \eta$, and (a) holds for $\KK' = \KK$. Vice versa, we require that
  $\KK \models \lnot \eta$ and obtain $\KK \models \iota \land \mu$ analogously to the proof of
  Proposition~\ref{prop:group}.
\end{proof}

\begin{example}
  \begin{enumerate}[(i)]
  \item Assume that $V(F)^*$ is a coset over $\CC$. Then $V(F)^*$ is a coset over
    $\RR$ if and only if $V(F)^* \neq \emptyset$ over $\RR$. This is the case for
    $F = \{x^2-2\}$ but not for $F = \{x^2+2\}$.
  \item Consider $F = \{x^4-4\} = \{(x^2-2)(x^2+2)\}$. Then over $\RR$,
    $V(F)^* = \{\pm \sqrt{2}\}$ is a coset, because
    $V(F)^*/\sqrt{2} = \{\pm 1\}$ is a group. Similarly over $\CC$,
    $V(F)^* = \{\pm \sqrt{2}, \pm i\sqrt{2}\}$ is a coset, as
    $V(F)^*/\sqrt{2} = \{\pm 1, \pm i\}$ is a group.
  \item Consider $F = \{x^4+x^2-6\} = \{(x^2-2)(x^2+3)\}$. Then over $\RR$,
    $V(F)^* = \{\pm \sqrt{2}\}$ is a coset, as $V(F)^*/\sqrt{2} = \{\pm 1\}$ is a
    group. Over $\CC$, in contrast, $V(F)^* = \{\pm \sqrt{2}, \pm i\sqrt{3}\}$ is not
    a coset.
  \end{enumerate}
\end{example}

\section{Hilbert's Nullstellensatz as a Swiss Army Knife}\label{se:hilbert}
A recent publication \cite{GrigorievIosif:19a} has systematically applied coset
tests to a large number for real-world models from the BioModels database
\cite{BioModels2015a}, investigating varieties over both the real and the
complex numbers. Over $\RR$ it used essentially our first-order sentences
presented in Section~\ref{se:firstorder} and applied efficient implementations
of real decision methods based on effective quantifier elimination
\cite{Weispfenning:88a,Weispfenning:97b,DolzmannSturm:97a,DolzmannSturm:97c,Seidl:06a,Kosta:16a}.

Over $\CC$, in contrast, it used a purely algebraic framework combining various
specialized methods from commutative algebra, typically based on Gröbner basis
computations \cite{Buchberger:65a,Faugere:99a}. This is in line with the vast
majority of the existing literature (cf.~the Introduction for references), which
uses computer algebra over algebraically closed fields, to some extent
supplemented with heuristic tests based on linear algebra.

Generalizing the successful approach for $\RR$ and aiming at a more uniform
overall framework, we want to study here the application of decision methods for
algebraically closed fields to our first-order sentences. Recall that our
sentences $\iota$, $\mu$, $\eta$, and $\gamma$ are universal $\LR$-sentences. Every such
sentence $\varphi$ can be equivalently transformed into a finite conjunction of
universal $\LR$-sentences of the following form:
\begin{displaymath}
  \textstyle
  \widehat \varphi ~\dotequal~ \forall x_1 \dots \forall x_n\biggl(
  \bigwedge\limits_{i=1}^m f_i(x_1, \dots, x_n) = 0 \longrightarrow g(x_1, \dots, x_n) = 0
  \biggr),
\end{displaymath}
where $f_1$, \dots, $f_m$, $g \in Z_p[x_1, \dots, x_n]$. Such $\LR$-sentences are called
uniform word problems \cite{BeckerWeispfenning:93a}. Over an algebraically
closed field $\bar \KK$ of characteristic $p$, Hilbert's Nullstellensatz
\cite{Hilbert:93a} provides a decision procedure for uniform word problems. It
states that
\begin{displaymath}
  \bar \KK \models \hilb \varphi \quad \text{if and only if} \quad g \in \sqrt{\langle f_1, \dots, f_m \rangle}.
\end{displaymath}
Recall that $\ACF_p$ is complete so that we furthermore have
$\ACF_p \models \hilb \varphi$ if and only if $\bar \KK \models \hilb \varphi$.

Our $\LR$-sentence $\iota$ for condition (\ref{eq:AAAinv}) can be equivalently
transformed into
\begin{multline*}
  \textstyle
  \forall g_1 \dots \forall g_n
  \forall x_1 \dots \forall x_n
  \biggl(
  \bigvee\limits_{j=1}^n g_j = 0 \lor
  \bigvee\limits_{j=1}^n x_j = 0 \lor{}\\
  \textstyle
  \bigvee\limits_{i=1}^m f_i(g_1,\dots,g_n) \neq 0 \lor
  \bigvee\limits_{i=1}^m f_i(g_1x_1,\dots,g_nx_n) \neq 0
  \lor \bigwedge\limits_{i=1}^m x^{d_i} f_i(g_1x_1^{-1},\dots,g_nx_n^{-1}) = 0
  \biggr),
\end{multline*}
which is in turn equivalent to
\begin{multline*}
  \hilb \iota ~\dotequal~ \textstyle
  \bigwedge\limits_{k=1}^m
  \forall g_1 \dots \forall g_n
  \forall x_1 \dots \forall x_n
  \biggl(
  \bigwedge\limits_{i=1}^m f_i(g_1,\dots,g_n)=0 \land
  \bigwedge\limits_{i=1}^m f_i(g_1x_1,\dots,g_nx_n) = 0\\
  \textstyle
  \longrightarrow x^{d_k} f_k(g_1x_1^{-1},\dots,g_nx_n^{-1}) \prod\limits_{j=1}^n g_jx_j = 0
  \biggr).
\end{multline*}
Hence, by Hilbert's Nullstellensatz, (\ref{eq:AAAinv}) holds in $\bar \KK$ if and
only if
\begin{equation}
  \label{eq:rinv}
  \textstyle
  x^{d_k} f_k(g_1x_1^{-1},\dots,g_nx_n^{-1}) \prod\limits_{j=1}^ng_jx_j \in R_{\ref{eq:AAAinv}}
  \quad \text{for all} \quad k \in M,
\end{equation}
where $R_{\ref{eq:AAAinv}} = \sqrt{\langle\,f_i(g_1,\dots,g_n), f_i(g_1x_1,\dots,g_nx_n) \mid i \in M \,\rangle}$.

Similarly, our $\LR$-sentence $\mu$ for condition (\ref{eq:AAAmult}) translates
into
\begin{multline*}
  \hilb \mu ~\dotequal~ \textstyle
  \bigwedge\limits_{k=1}^m
  \forall g_1 \dots \forall g_n
  \forall x_1 \dots \forall x_n
  \forall y_1 \dots \forall y_n\\
  \textstyle
  \biggl(
  \bigwedge\limits_{i=1}^m f_i(g_1,\dots,g_n)=0 \land
  \bigwedge\limits_{i=1}^m f_i(g_1x_1,\dots,g_nx_n) = 0 \land 
  \bigwedge\limits_{i=1}^m f_i(g_1y_1,\dots,g_ny_n) = 0\\
  \textstyle
  \longrightarrow f_k(g_1x_1y_1,\dots,g_nx_ny_n) \prod\limits_{j=1}^n g_jx_jy_j = 0
  \biggr).
\end{multline*}
Again, by Hilbert's Nullstellensatz, (\ref{eq:AAAmult}) holds in $\bar \KK$ if and
only if
\begin{equation}
  \label{eq:rmult}
  \textstyle
  f_k(g_1x_1y_1, \dots, g_nx_ny_n) \prod\limits_{j=1}^n g_jx_jy_j \in R_{\ref{eq:AAAmult}}
  \quad \text{for all} \quad k\in M,
\end{equation}
where $R_{\ref{eq:AAAmult}} = \sqrt{\langle\,f_i(g), f_i(gx), f_i(gy) \mid i \in M\,\rangle}$.

Next, our $\LR$-sentence $\eta$ is is equivalent to
\begin{displaymath}
  \textstyle
  \hilb \eta ~\dotequal~ \forall x_1 \dots \forall x_n\biggr(
  \bigwedge\limits_{i=1}^m f_i = 0 \longrightarrow \prod\limits_{j=1}^n x_j = 0
  \biggl).
\end{displaymath}
Using once more Hilbert's Nullstellensatz, $\bar \KK \models \hilb \eta$ if and only if
\begin{equation}\label{eq:empty}
  \textstyle
  \prod\limits_{j=1}^n x_j \in R_{\ref{eq:AAAne}},
\end{equation}
where $R_{\ref{eq:AAAne}} = \sqrt{\langle f_1, \dots, f_m \rangle}$. Hence our non-emptiness
condition~(\ref{eq:AAAne}) holds in $\bar \KK$ if and only if
\begin{equation}\label{eq:nonempty}
  \textstyle
  \prod\limits_{j=1}^n x_j \notin R_{\ref{eq:AAAne}}.
\end{equation}

Finally, our $\LR$-sentence $\gamma$ for condition~(\ref{eq:group}) is equivalent
to
\begin{displaymath}
  \textstyle
  \hilb \gamma ~\dotequal~ \bigwedge\limits_{k=1}^m \bigl(0 = 0 \longrightarrow f_k(1, \dots, 1) = 0\bigr).
\end{displaymath}
Here Hilbert's Nullstellensatz tells us that condition~(\ref{eq:group}) holds in
$\bar \KK$ if and only if
\begin{equation}
  \label{eq:group2}
  f_k(1, \dots, 1) \in R_{\ref{eq:group}} \quad \text{for all} \quad k \in M,
\end{equation}
where $R_{\ref{eq:group}} = \sqrt{\langle 0 \rangle} = \langle 0 \rangle$. Notice that the radical
membership test quite naturally reduces to the obvious test with plugging in.

\section{Complexity}\label{se:complexity}
Let us briefly discuss asymptotic complexity bounds around problems and methods
addressed here. We do so very roughly, in terms of the input word length. The
cited literature provides more precise bounds in terms of several complexity
parameters, such as numbers of quantifiers, or degrees.

The decision problem for algebraically closed fields is double exponential
\cite{Heintz:83a} in general, but only single exponential when the number of
quantifier alternations is bounded \cite{Grigorev:87a}, which covers in
particular our universal formulas. The decision problem for real closed fields
is double exponential as well \cite{DavenportHeintz:88a}, even for linear
problems \cite{Weispfenning:88a}; again it becomes single exponential when
bounding the number of quantifier alternations \cite{Grigoriev:88a}.

Ideal membership tests are at least double exponential \cite{MAYR1982305}, and
it was widely believed that this would impose a corresponding lower bound also
for any algorithm for Hilbert's Nullstellensatz. Quite surprisingly, it turned
out that there are indeed single exponential such algorithms
\cite{10.2307/1971361,10.2307/1990996}.

On these grounds it is clear that our coset tests addressed in the previous
sections can be solved in single exponential time for algebraically closed
fields as well as for real closed fields. Recall that our considering those
tests is actually motivated by our interest in shifted toricity, which requires,
in addition, the irreducibility of the considered variety over the considered
domain. Recently it has been shown that testing shifted toricity, including
irreducibility, is also only single exponential over algebraically closed fields
as well as real closed fields \cite{GrigorievIosif:19a}.

Most asymptotically fast algorithms mentioned above are not implemented and it
is not clear that they would be efficient in practice.

\begin{figure}
  \lstset{language=Maple, basicstyle=\scriptsize, basewidth=0.5em,
    numbers=right, aboveskip=0pt, belowskip=0pt}

\begin{lstlisting}
ToricHilbert := proc(F::list(polynom))
uses PolynomialIdeals;

local Iota := proc()::truefalse;
local R1, s, prod, f;
   R1 := < op(subs(zip(`=`, xl, gl), F)), op(subs(zip((x, g) -> x = g*x, xl, gl), F)) >;
   s := zip((x, g) -> x = g/x, xl, gl);
   prod := g * x;
   for f in subs(s, F) do
      if not RadicalMembership(numer(f) * prod, R1) then
         return false
      end if
   end do;
   return true
end proc;

local Mu := proc()::truefalse;
local R2, s, prod, f;
   R2 := < op(subs(zip(`=`, xl, gl), F)), op(subs(zip((x, g) -> x = g*x, xl, gl), F)),
           op(subs(zip(`=`, xl, zip(`*`, gl, yl)), F)) >;
   s := zip(`=`, xl, zip(`*`, gl, zip(`*`, xl, yl)));
   prod := g * x * y;
   for f in subs(s, F) do
      if not RadicalMembership(f * prod, R2) then
         return false
      end if
   end do;
   return true
end proc;

local Eta := proc()::truefalse;
local R3, prod;
   R3 := < op(F) >;
   prod := foldl(`*`, 1, op(xl));
   return RadicalMembership(prod, R3)
end proc;

local Gamma := proc()::truefalse;
local R4, s, f;
   R4 := < 0 >;
   s := map(x -> x=1, xl);
   for f in subs(s, F) do
      if not RadicalMembership(f, R4) then
         return false
      end if
   end do;
   return true
end proc;

local Rename := proc(base::name, l::list(name))::list(name);
uses StringTools;
   return map(x -> cat(base, Select(IsDigit, x)), l)
end proc;

local X, xl, gl, yl, g, x, y, iota, t_iota, mu, t_mu, eta, t_eta, gamma_, t_gamma, coset, group, t;
   t := time();
   xl := convert(indets(F), list);
   x := foldl(`*`, 1, op(xl));
   gl := Rename('g', xl);
   g := foldl(`*`, 1, op(gl));
   yl := Rename('y', xl);
   y := foldl(`*`, 1, op(yl));
   t_iota := time(); iota := Iota(); t_iota := time() - t_iota;
   t_mu := time(); mu := Mu(); t_mu := time() - t_mu;
   t_eta := time(); eta := Eta(); t_eta := time() - t_eta;
   t_gamma := time(); gamm := Gamma(); t_gamma := time() - t_gamma;
   coset := iota and mu and not eta;
   group := iota and mu and gamm;
   t := time() - t;
   return nops(F), nops(xl), iota, t_iota, mu, t_mu, eta, t_eta, gamm, t_gamma, coset, group, t
end proc;
\end{lstlisting}
  \caption{Maple code for computing one row of
    Table~\ref{tab:computations}\label{fig:maple}}
\end{figure}

\section{Computational Experiments}\label{se:computations}
We have studied 129 models from the
BioModels\footnote{\url{https://www.ebi.ac.uk/biomodels/}} database
\cite{BioModels2015a}. Technically, we took our input from
ODEbase\footnote{\url{http://odebase.cs.uni-bonn. de/}} which provides
preprocessed versions for symbolic computation. Our 129 models establish the
complete set currently provided by ODEbase for wich the relevant systems of
ordinary differential equations have polynomial vector fields.

We limited ourselves to characteristic 0 and applied the tests (\ref{eq:rinv}),
(\ref{eq:rmult}), (\ref{eq:empty}), (\ref{eq:group2}) derived in
Section~\ref{se:hilbert} using Hilbert's Nullstellensatz. Recall that those
tests correspond to $\iota$, $\mu$, $\eta$, $\gamma$ from Section~\ref{se:generalization},
respectively, and that one needs $\iota \land \mu \land \lnot \eta$ or
$\iota \land \mu \land \gamma$ for cosets or groups, respectively. From a symbolic computation
point of view, we used exclusively polynomial arithmetic and radical membership
test. The complete Maple code for computing a single model is displayed in
Figure~\ref{fig:maple}; it is surprisingly simple.

\afterpage{%
  \thispagestyle{empty}
  \begin{table}
    \centering
    \caption{Results and computation times (in seconds) of our computations on
      models from the BioModels database
      \cite{BioModels2015a}\label{tab:computations}}
    \medskip
    
    \scriptsize
    \begin{tabular}{@{}lrrlrlrlrlrllr@{}}
      \hline
      model  &  $m$  &  $n$  &  $\iota$  &  $t_\iota$  &  $\mu$  &  $t_\mu$  &  $\eta$  &  $t_\eta$ &  $\gamma$  &  $t_\gamma$  &  coset  &  group  &  $t_\Sigma$\\
      \hline
      001  &  12  &  12  &  true  &  7.826  &  true  &  7.86  &  false  &  4.267  &  false  &  0.053  &  true  &  false  &  20.007\\
      040  &  5  &  3  &  false  &  1.415  &  false  &  0.173  &  false  &  0.114  &  false  &  0.043  &  false  &  false  &  1.746\\
      050  &  14  &  9  &  true  &  1.051  &  true  &  2.458  &  true  &  0.113  &  false  &  0.05  &  false  &  false  &  3.673\\
      052  &  11  &  6  &  true  &  3.605  &  true  &  1.635  &  true  &  0.096  &  false  &  0.059  &  false  &  false  &  5.396\\
      057  &  6  &  6  &  true  &  0.271  &  true  &  0.263  &  false  &  0.858  &  false  &  0.045  &  true  &  false  &  1.438\\
      072  &  7  &  7  &  true  &  0.763  &  true  &  0.496  &  true  &  0.08  &  false  &  0.06  &  false  &  false  &  1.4\\
      077  &  8  &  7  &  true  &  0.296  &  true  &  0.356  &  false  &  0.097  &  false  &  0.051  &  true  &  false  &  0.801\\
      080  &  10  &  10  &  true  &  0.714  &  true  &  1.341  &  true  &  0.103  &  false  &  0.06  &  false  &  false  &  2.219\\
      082  &  10  &  10  &  true  &  0.384  &  true  &  0.39  &  true  &  0.086  &  false  &  0.041  &  false  &  false  &  0.902\\
      091  &  16  &  14  &  true  &  0.031  &  true  &  0.045  &  true  &  0.003  &  false  &  0.062  &  false  &  false  &  0.142\\
      092  &  4  &  3  &  true  &  0.293  &  true  &  0.244  &  false  &  0.104  &  false  &  1.03  &  true  &  false  &  1.671\\
      099  &  7  &  7  &  true  &  0.298  &  true  &  0.698  &  false  &  0.087  &  false  &  0.036  &  true  &  false  &  1.119\\
      101  &  6  &  6  &  false  &  4.028  &  false  &  10.343  &  false  &  0.917  &  false  &  0.073  &  false  &  false  &  15.361\\
      104  &  6  &  4  &  true  &  0.667  &  true  &  0.146  &  true  &  0.084  &  false  &  0.039  &  false  &  false  &  0.937\\
      105  &  39  &  26  &  true  &  0.455  &  true  &  0.367  &  true  &  0.043  &  false  &  0.038  &  false  &  false  &  0.905\\
      125  &  5  &  5  &  false  &  0.193  &  false  &  0.098  &  false  &  0.078  &  false  &  0.038  &  false  &  false  &  0.408\\
      150  &  4  &  4  &  true  &  0.173  &  true  &  0.153  &  false  &  0.094  &  false  &  0.043  &  true  &  false  &  0.464\\
      156  &  3  &  3  &  true  &  2.638  &  true  &  0.248  &  false  &  0.86  &  false  &  0.052  &  true  &  false  &  3.8\\
      158  &  3  &  3  &  false  &  0.148  &  false  &  0.149  &  false  &  0.16  &  false  &  0.045  &  false  &  false  &  0.503\\
      159  &  3  &  3  &  true  &  0.959  &  true  &  0.175  &  false  &  0.083  &  false  &  0.04  &  true  &  false  &  1.257\\
      178  &  6  &  4  &  true  &  0.52  &  true  &  1.71  &  true  &  0.877  &  false  &  1.201  &  false  &  false  &  4.308\\
      186  &  11  &  10  &  true  &  31.785  &  true  &  1026.464  &  true  &  1.956  &  false  &  0.095  &  false  &  false  &  1060.301\\
      187  &  11  &  10  &  true  &  27.734  &  true  &  1023.648  &  true  &  0.103  &  false  &  0.062  &  false  &  false  &  1051.548\\
      188  &  20  &  10  &  true  &  0.075  &  true  &  0.079  &  true  &  0.04  &  false  &  0.047  &  false  &  false  &  0.242\\
      189  &  18  &  7  &  true  &  0.035  &  true  &  0.02  &  true  &  0.002  &  false  &  0.062  &  false  &  false  &  0.12\\
      194  &  5  &  5  &  false  &  2.338  &  false  &  1.922  &  false  &  0.612  &  false  &  0.05  &  false  &  false  &  4.922\\
      197  &  7  &  5  &  false  &  7.562  &  false  &  71.864  &  false  &  0.485  &  false  &  0.05  &  false  &  false  &  79.962\\
      198  &  12  &  9  &  true  &  0.397  &  true  &  0.793  &  true  &  0.077  &  false  &  0.042  &  false  &  false  &  1.31\\
      199  &  15  &  8  &  true  &  1.404  &  true  &  1.531  &  false  &  0.215  &  false  &  0.054  &  true  &  false  &  3.205\\
      220  &  58  &  56  &  true  &  146.146  &  true  &  534.832  &  true  &  6.921  &  false  &  0.964  &  false  &  false  &  688.866\\
      227  &  60  &  39  &  true  &  0.273  &  true  &  0.485  &  true  &  0.01  &  false  &  0.077  &  false  &  false  &  0.847\\
      229  &  7  &  7  &  true  &  1.917  &  true  &  3.348  &  false  &  0.131  &  false  &  0.062  &  true  &  false  &  5.458\\
      233  &  4  &  2  &  false  &  0.16  &  false  &  0.44  &  false  &  0.17  &  false  &  0.557  &  false  &  false  &  1.328\\
      243  &  23  &  19  &  true  &  8.598  &  true  &  1171.687  &  true  &  2.512  &  false  &  0.171  &  false  &  false  &  1182.97\\
      259  &  17  &  16  &  true  &  1.334  &  true  &  1.913  &  true  &  0.092  &  false  &  0.045  &  false  &  false  &  3.385\\
      260  &  17  &  16  &  true  &  2.182  &  true  &  0.748  &  true  &  0.079  &  false  &  0.047  &  false  &  false  &  3.057\\
      261  &  17  &  16  &  true  &  3.359  &  true  &  2.872  &  true  &  0.113  &  false  &  0.095  &  false  &  false  &  6.44\\
      262  &  11  &  9  &  true  &  0.402  &  true  &  0.41  &  true  &  0.091  &  false  &  0.071  &  false  &  false  &  0.975\\
      263  &  11  &  9  &  true  &  0.379  &  true  &  0.403  &  true  &  0.085  &  false  &  0.066  &  false  &  false  &  0.934\\
      264  &  14  &  11  &  true  &  1.031  &  true  &  2.036  &  true  &  0.136  &  false  &  0.063  &  false  &  false  &  3.268\\
      267  &  4  &  3  &  true  &  1.084  &  true  &  0.246  &  true  &  0.095  &  false  &  0.049  &  false  &  false  &  1.475\\
      271  &  6  &  4  &  true  &  0.286  &  true  &  0.283  &  true  &  0.746  &  false  &  0.045  &  false  &  false  &  1.361\\
      272  &  6  &  4  &  true  &  0.361  &  true  &  0.323  &  true  &  0.086  &  false  &  0.055  &  false  &  false  &  0.826\\
      281  &  32  &  32  &  true  &  20.987  &  true  &  29.791  &  true  &  0.602  &  false  &  0.055  &  false  &  false  &  51.437\\
      282  &  6  &  3  &  true  &  0.205  &  true  &  0.19  &  true  &  0.087  &  false  &  0.046  &  false  &  false  &  0.528\\
      283  &  4  &  3  &  true  &  0.294  &  true  &  0.211  &  true  &  0.087  &  false  &  0.412  &  false  &  false  &  1.005\\
      289  &  5  &  4  &  false  &  2.291  &  false  &  1.118  &  false  &  0.165  &  false  &  0.044  &  false  &  false  &  3.619\\
      292  &  6  &  2  &  true  &  0.06  &  true  &  0.048  &  true  &  0.063  &  false  &  0.046  &  false  &  false  &  0.218\\
      306  &  5  &  2  &  true  &  0.149  &  true  &  0.121  &  false  &  0.079  &  false  &  0.041  &  true  &  false  &  0.391\\
      307  &  5  &  2  &  true  &  0.129  &  true  &  0.121  &  true  &  0.043  &  false  &  0.148  &  false  &  false  &  0.441\\
      310  &  4  &  1  &  true  &  0.053  &  true  &  0.369  &  true  &  0.047  &  false  &  0.04  &  false  &  false  &  0.509\\
      311  &  4  &  1  &  true  &  0.076  &  true  &  0.048  &  true  &  0.224  &  false  &  0.048  &  false  &  false  &  0.397\\
      312  &  3  &  2  &  true  &  0.098  &  true  &  0.512  &  true  &  0.043  &  false  &  0.043  &  false  &  false  &  0.697\\
      314  &  12  &  10  &  true  &  0.515  &  true  &  1.789  &  true  &  0.1  &  false  &  0.059  &  false  &  false  &  2.464\\
      321  &  3  &  3  &  true  &  0.163  &  true  &  0.148  &  true  &  0.042  &  false  &  0.039  &  false  &  false  &  0.393\\
      357  &  9  &  8  &  true  &  0.353  &  true  &  1.517  &  true  &  0.07  &  false  &  0.045  &  false  &  false  &  1.986\\
      359  &  9  &  8  &  true  &  1.677  &  true  &  3.605  &  true  &  0.11  &  false  &  0.055  &  false  &  false  &  5.448\\
      360  &  9  &  8  &  true  &  0.479  &  true  &  0.47  &  true  &  0.096  &  false  &  0.05  &  false  &  false  &  1.096\\
      361  &  8  &  8  &  true  &  1.069  &  true  &  2.746  &  true  &  0.156  &  false  &  0.045  &  false  &  false  &  4.017\\
      363  &  4  &  3  &  true  &  0.244  &  true  &  0.199  &  true  &  0.077  &  false  &  0.041  &  false  &  false  &  0.561\\
      364  &  14  &  12  &  true  &  2.483  &  true  &  7.296  &  true  &  0.55  &  false  &  0.064  &  false  &  false  &  10.394\\
      413  &  5  &  5  &  false  &  1.55  &  false  &  22.323  &  false  &  0.117  &  false  &  0.053  &  false  &  false  &  24.044\\
      459  &  4  &  3  &  true  &  0.542  &  true  &  0.224  &  false  &  0.18  &  false  &  0.068  &  true  &  false  &  1.014\\
      460  &  4  &  3  &  false  &  1.025  &  false  &  0.936  &  false  &  0.143  &  false  &  0.216  &  false  &  false  &  2.321\\
      475  &  23  &  22  &  true  &  97.876  &  true  &  3377.021  &  true  &  0.231  &  false  &  0.062  &  false  &  false  &  3475.192\\
      484  &  2  &  1  &  true  &  0.384  &  true  &  0.143  &  false  &  0.099  &  false  &  0.048  &  true  &  false  &  0.674\\
      485  &  2  &  1  &  false  &  0.564  &  false  &  0.354  &  false  &  0.209  &  false  &  0.042  &  false  &  false  &  1.169\\
      486  &  2  &  2  &  true  &  0.119  &  true  &  0.106  &  false  &  0.073  &  false  &  0.041  &  true  &  false  &  0.339\\
      487  &  6  &  6  &  true  &  0.475  &  true  &  1.008  &  false  &  0.099  &  false  &  0.045  &  true  &  false  &  1.628\\
      491  &  57  &  57  &  true  &  123.138  &  true  &  536.865  &  false  &  2.067  &  true  &  0.007  &  true  &  true  &  662.08\\
      492  &  52  &  52  &  true  &  85.606  &  true  &  284.753  &  false  &  1.123  &  true  &  0.003  &  true  &  true  &  371.489\\
      519  &  3  &  3  &  true  &  1.357  &  true  &  2.367  &  false  &  5.142  &  false  &  0.097  &  true  &  false  &  8.964\\
      546  &  7  &  3  &  true  &  0.327  &  true  &  0.338  &  true  &  0.109  &  false  &  0.042  &  false  &  false  &  0.817\\
      559  &  90  &  71  &  true  &  4.742  &  true  &  7.525  &  true  &  0.19  &  false  &  0.053  &  false  &  false  &  12.515\\
      584  &  35  &  9  &  true  &  0.4  &  true  &  0.655  &  false  &  0.095  &  false  &  0.043  &  true  &  false  &  1.194\\
      619  &  10  &  8  &  true  &  0.411  &  true  &  0.443  &  true  &  0.087  &  false  &  0.052  &  false  &  false  &  0.994\\
      629  &  5  &  5  &  true  &  0.209  &  true  &  0.197  &  false  &  0.079  &  false  &  0.046  &  true  &  false  &  0.532\\
      647  &  11  &  11  &  false  &  0.854  &  false  &  16.436  &  false  &  0.165  &  false  &  0.051  &  false  &  false  &  17.507\\
      \hline
    \end{tabular}
  \end{table}
\clearpage
}

We conducted our computations on a 2.40 GHz Intel Xeon E5-4640 with 512 GB RAM
and 32 physical cores providing 64 CPUs via hyper-threading. For parallelization
of the jobs for the individual models we used GNU Parallel \cite{Tange:11a}.
Results and timings are collected in Table~\ref{tab:computations}. With a time
limit of one hour CPU time per model we succeeded on 78 models, corresponding to
60\%, the largest of which, no.~559, has 90 polynomials in 71 dimensions. The
median of the overall computation times for the successful models is 1.419~s. We
would like to emphasize that our focus here is illustrating and evaluating our
overall approach, rather than obtaining new insights into the models. Therefore
our code in Figure~\ref{fig:maple} is very straightforward without any
optimizations. In particular, computation continues even when one relevant
subtest has already failed. More comprehensive results on our dataset can be
found in \cite{GrigorievIosif:19a}.

Among our 78 successfully computed models, we detected 20 coset cases,
corresponding to 26\%. Two out of those 20 are even group cases. Among the 58
other cases, 46, corresponding to 78\%, fail only due to their emptiness $\eta$; we
know from \cite{GrigorievIosif:19a} that many such cases exhibit in fact coset
structure when considered in suitable lower-dimensional spaces, possibly after
prime decomposition. Finally notice that our example reaction from the
Introduction, no.~92, is among the smallest ones with a coset structure.

\section{Conclusions and Future Work}\label{se:conclusions}
We have used Hilbert's Nullstellensatz to derive important information about the
varieties of biological models with a polynomial vector field $F$. The key
technical idea was generalizing from pure algebra to more general first-order
logic. Recall from Section~\ref{se:generalization} that except for non-emptiness
of $V(F)^*$ the information we obtained is valid in \emph{all fields} of
characteristic 0. Wherever we discovered non-emptiness, this holds at least in
\emph{all algebraically closed fields} of characteristic 0. For transferring our
obtained results to real closed fields, e.g., subtropical methods
\cite{Sturm:15b,FontaineOgawa:17b,HongSturm:18a} provide fast heuristic tests
for the non-emptiness of $V(F)^*$ there.

Technically, we only used polynomial arithmetic and polynomial radical
membership tests. This means that on the software side there are many
off-the-shelf computer algebra systems and libraries available where our ideas
could be implemented, robustly and with little effort. This in turn makes it
attractive for the integration with software from systems biology, which could
open exciting new perspectives for symbolic computation with applications
ranging from the fundamental research in the life sciences to state-of-the-art
applied research in medicine and pharmacology.

We had motivated our use of Hilbert's Nullstellensatz by viewing it as a
decision procedure for the universal fragment of first-order logic in
algebraically closed fields, which is sufficient for our purposes. Our focus on
algebraically closed fields here is in accordance with the majority of existing
literature on toricity. However, it is generally accepted that from a scientific
point of view, real closed fields are the appropriate domain to consider.

We have seen in Section~\ref{se:complexity} that the theoretical complexities
for general decision procedures in algebraically closed fields vs.~real closed
fields strongly resemble each other. What could now take the place of Hilbert's
Nullstellensatz over the reals with respect to practical computations on model
sizes as in Table~\ref{tab:computations} or even larger? A factor of 10 could
put us in the realm of models currently used in the development of drugs for
diabetes or cancer. One possible answer is \emph{satisfiability modulo theories
  solving (SMT)} \cite{DBLP:journals/jacm/NieuwenhuisOT06}.\footnote{SMT
  technically aims at the existential fragment, which in our context is
  equivalent to the universal fragment via logical negation.} SMT is incomplete
in the sense that it often proves or disproves validity, but it can yield
``unknown'' for specific input problems. When successful, it is typically
significantly faster than traditional algebraic decision procedures. For coping
with incompleteness one can still fall back into real quantifier elimination.
Interest in collaboration between the SMT and the symbolic computation
communities exists on both sides \cite{AbrahamAbbott:2016a,AbrahamAbbott:2016b}.

\subsubsection*{Acknowledgments}
This work has been supported by the interdisciplinary bilateral project
ANR-17-CE40-0036/DFG-391322026 SYMBIONT
\cite{BoulierFages:18a,BoulierFages:18b}. We are grateful to Dima Grigoriev for
numerous inspiring and very constructive discussions around toricity.

\end{document}